\documentclass[twocolumn,pra,superscriptaddress]{revtex4-1}

\usepackage{qcircuit}
\usepackage[dvips]{graphicx}
\usepackage{amsmath,amssymb,amsthm,mathrsfs,amsfonts,dsfont}
\usepackage{subfigure, epsfig}
\usepackage{braket}
\usepackage{bm}
\usepackage{enumerate}
\usepackage{color}
\usepackage{graphicx}
\usepackage{algorithm}
\usepackage[noend]{algpseudocode}

\usepackage{mathdots}
\usepackage{bm}
\usepackage{nicefrac}

\usepackage{algorithm}
\usepackage[noend]{algpseudocode}

\newtheorem{theorem}{Theorem}
\newtheorem{lemma}{Lemma}

\newcommand{\inn}[2]{\langle #1,#2\rangle}
\newcommand{\tr}{\mathrm{Tr}}

\begin{document}

\title{Efficient and robust detection of multipartite Greenberger-Horne-Zeilinger-like states}
\date{\today}

\author{Qi Zhao}

\affiliation{Center for Quantum Information, Institute for Interdisciplinary Information Sciences, Tsinghua University, Beijing, 100084 China}
\author{Gerui Wang}
\affiliation{Department of Computer Science, University of Illinois at Urbana-Champaign, Urbana, IL 61801, USA}

\author{Xiao Yuan}
\email{xiao.yuan.ph@gmail.com}
\affiliation{Department of Materials, University of Oxford, Parks Road, Oxford OX1 3PH, United Kingdom}

\author{Xiongfeng Ma}
\email{xma@tsinghua.edu.cn}
\affiliation{Center for Quantum Information, Institute for Interdisciplinary Information Sciences, Tsinghua University, Beijing, 100084 China}

\begin{abstract}
Entanglement is a key resource for quantum information processing. A widely used tool for detecting entanglement is entanglement witness, where the measurement of the witness operator is guaranteed to be positive for all separable states and can be negative for certain entangled states. In reality, due to the exponentially increasing the Hilbert-space dimension with respective to the system size, it is very challenging to construct an efficient entanglement witness for general multipartite entangled states. For $N$-partite Greenberger-Horne-Zeilinger (GHZ)-like states, the most robust witness scheme requires $N+1$ local measurement settings and can tolerate up to $1/2$ white noise. As a comparison, the most efficient witness for GHZ-like states only needs two local measurement settings and can tolerate up to $1/3$ white noise. There is a trade-off between the realization efficiency, the number of measurement settings, and the detection robustness, the maximally tolerable white noise. In this work, we study this trade-off by proposing a family of entanglement witnesses with $k$ ($2\le k\le N+1$) local measurement settings. Considering symmetric local measurements, we calculate the maximal tolerable noise for any given number of measurement settings. Consequently, we design the optimal witness with a minimal number of settings for any given level of white noise. Our theoretical analysis can be applied to other multipartite entangled states with a strong symmetry. Our witnesses can be easily implemented in experiment and applied in practical multipartite entanglement detection under different noise conditions.
\end{abstract}

\maketitle

\section{Introduction}
Quantum mechanics exhibits a variety of counter-intuitive phenomena. Among them, quantum entanglement , which enables the ``spooky actions" between two space-likely separated parties, has been intensively studied \cite{Horodecki09}. It is widely believed that entanglement plays an essential role in many quantum information processing tasks, such as quantum computation \cite{Shor97}, quantum key distribution  \cite{bb84,PhysRevLett.67.661}, quantum teleportation \cite{Bennett1993Teleport}, and Bell's inequality test \cite{bell,greenberger1990bell,ansmann2009violation}. A bipartite state is entangled when it cannot be represented as a mixture of product states $\rho_A^i\otimes\rho_B^i$, that is, $\rho=\sum_i p_i \rho_A^i\otimes\rho_B^i$. This definition can be naturally extended to the multipartite scenario \cite{RevModPhys.81.865,guhne2009entanglement}. An $N$-party state is genuinely entangled when it cannot be expressed as a mixture of bi-separable states $\rho_{\{A\}}^i\otimes\rho_{\{B\}}^i$ such that $\rho=\sum_i p_i \rho_{\{A\}}^i\otimes\rho_{\{B\}}^i$ for any bipartition $\{A\}$ and $\{B\}$ of $N$ parties.

Great efforts have been made to theoretically \cite{RevModPhys.81.865,guhne2009entanglement,PhysRevLett.94.060501} and experimentally \cite{PhysRevLett.106.130506, PhysRevLett.117.210502, PhysRevLett.112.155304,chen2017observation,PhysRevLett.120.260502} witness the genuine multipartite entanglement. In general, entanglement detection of an arbitrary unknown state is a challenging task. To determine whether a state is entangled or not, the entanglement witness (EW) method is widely employed, taking advantage of the convexity of the set of separable states \cite{RevModPhys.81.865,horodecki1996separability}. A Hermitian witness operator $\mathcal{W}$ is called an EW when for any separable state $\sigma\in \mathcal{S}_{sep}$, $\mathrm{Tr}(\mathcal{W}\sigma)\ge 0$, and for certain entangled state $\rho$, $\mathrm{Tr}(\mathcal{W}\rho)<0$. In general, a witness $\mathcal{W}$ can be represented as
\begin{equation}\label{generalwitness}
\begin{aligned}
\mathcal{W}=\alpha \mathbb{I}- \mathcal{O},
\end{aligned}
\end{equation}
where $\mathbb{I}$ is the identity operator, $\mathcal{O}$ is a Hermitian operator, and the coefficient $\alpha$ is determined by
\begin{equation}\label{eq:defalpha}
\begin{aligned}
\alpha=\sup\{\mathrm{Tr}(\mathcal{O}\sigma): \sigma\in \mathcal{S}_{sep}\}.
\end{aligned}
\end{equation}
Choosing the observable $\mathcal{O}$ and determining $\alpha$ are two important but difficult tasks. A witness can be constructed by various methods, for instance, positive partial transpose criterion \cite{peres1996separability} and computable cross norm or realignment criterion  \cite{rudolph2005further,PhysRevLett.95.150504}. We refer to \cite{RevModPhys.81.865,guhne2009entanglement,RevModPhys.82.277,plbnio2007introduction} for detailed review of quantum entanglement detection.

In this work, we mainly focus on the genuine entanglement detection of multipartite Greenberger-Horne-Zeilinger (GHZ)-like states \cite{greenberger1990bell}, which has been employed in many quantum information processing tasks, such as quantum error correction \cite{nielsen2002quantum}, multipartite nonlocality tests \cite{scarani2001spectral}, and multipartite quantum key agreement \cite{chen2004multi}. For an $N$-partite GHZ state \cite{greenberger1990bell},
\begin{equation}\label{eq:NGHZ}
\begin{aligned}
\ket{GHZ}=\frac{1}{\sqrt{2}}\left(\ket{0}^{\otimes N}+\ket{1}^{\otimes N}\right),
\end{aligned}
\end{equation}
a conventional witness is to take $\mathcal{O}=\ket{GHZ}\bra{GHZ}$ and construct a witness following Eq.~\eqref{generalwitness},
\begin{equation} \label{eq:NEW0}
\begin{aligned}
\mathcal{W}=\frac{1}{2} \mathbb{I}- \ket{GHZ}\bra{GHZ}.
\end{aligned}
\end{equation}
Such a witness is widely used in experimental multipartite entangled states detection \cite{PhysRevLett.117.210502,chen2017observation, PhysRevX.8.021072}, especially for a practical GHZ state that is mixed with white noises,
\begin{equation} \label{eq:GHZnoise}
\begin{aligned}
\rho = (1-p)\ket{GHZ}\bra{GHZ}+ p2^{-N}\mathbb{I}.
\end{aligned}
\end{equation}
Such a noisy GHZ state is genuinely entangled when $p\lesssim 1/2$ for a large $N$.

In practice, a multipartite EW is often implemented by decomposing it into several local measurements, and each party performs the corresponding measurements. The witness shown in Eq.~\eqref{eq:NEW0} can be realized by measuring $N+1$ local measurement settings (LMS) \cite{guhne2007toolbox}. In the meantime, there also exists a witness that is based on the stabilizer of the GHZ state and only requires two LMS, for any number of parties $N$ \cite{PhysRevLett.94.060501,PhysRevLett.117.210504}. However, this witness can only detect the genuine entanglement of noisy GHZ states defined in Eq.~\eqref{eq:GHZnoise} with $p\lesssim 1/3$ for a large $N$.

In experiment, under the requirement of a successful detection of entanglement, a smaller number of LMSs is highly desired. This is because a LMS generally requires at least thousand rounds of experiments in order to obtain an accurate estimation of the expectation value with relatively small statistical fluctuations. Thus, the witness requiring $N+1$ LMSs will take more efforts, especially for a large $N$.
When the noise level is low, one can simply implement the witness with two LMSs. However, it may be experimentally hard to keep the noise level low. When the noise level increases to a certain amount, the witness with two LMSs is more likely to fail and additional measurements must be employed. However, besides the two extreme witnesses with two and $N+1$ LMSs, the witnesses with $k$ ($2<k\le N$) LMSs still lacks studying. These intermediate witnesses will provide alternative options for finding the efficient and valid witness with a given noise level, which is instructive for experimental entanglement detection of GHZ-like states. In a sense, there exists a trade-off between the detection robustness, i.e., the maximal tolerable white noise, and the realization efficiency, i.e., the number of LMSs. It is also theoretically interesting to study such a trade-off for multipartite entanglement detection.

In this work, we first propose a systematic method of constructing a family of witnesses with $k$ ($2\le k\le N+1$) LMSs in Section \ref{Sc:kfamily}. We show how to calculate the witness, specifically, determining the coefficient $\alpha$ as defined in Eq.~\eqref{generalwitness}, by converting the problem into solving a simple maximization problem. In Section \ref{Sec:robust}, we numerically solve the maximization and study the optimal witness under different noise levels. That is, we explicitly give the witness that have the fewest LMSs and are able to detect noisy $N$-partite GHZ states with certain amount of noise levels for $N=5$, $10$, and $15$. We verify the trade-off between detection robustness and realization efficiency. In Section \ref{Sec:Conclusion}, we conclude our results and discuss possible applications.

\section{EW family with $k$ LMS} \label{Sc:kfamily}
In this section, we first review two EWs that requires 2 and $N+1$ LMSs. Then, we propose a family of EWs that contain $k$ ($2\le k\le N+1$) settings. We formulate the key steps in constructing the witnesses in the first subsection and give the detailed derivation for witnesses in the next.

\subsection{Constructing EW family}
The simplest multipartite EW consists of 2 LMSs \cite{PhysRevLett.94.060501}, only utilizing $\sigma_z$ and $\sigma_x$ measurements,
\begin{equation}\label{witness1}
\begin{aligned}
\mathcal{W}_2=\frac{1}{2} \mathbb{I}- \frac{1}{2}M_z-\frac{1}{4}M_x,
\end{aligned}
\end{equation}
where $M_z=\ket{0}\bra{0}^{\otimes{N}}+\ket{1}\bra{1}^{\otimes{N}}$ and $M_x=\sigma_x^{\otimes{N}}$. Here, $\sigma_x=\ket{0}\bra{1}+\ket{1}\bra{0}$, $\sigma_y=-i\ket{0}\bra{1}+i\ket{1}\bra{0}$, and $\sigma_z=\ket{0}\bra{0}-\ket{1}\bra{1}$ are the Pauli matrices. Moreover, this witness with a modified coefficient, $c M_z + M_x$ can detect not only genuine entanglement but also the separability of multipartite GHZ-like states \cite{PhysRevX.8.021072}.

In reality, one can also decompose the widely used multipartite EW in Eq.~\eqref{eq:NEW0} into several local observables and each party performs the corresponding measurements. With $N+1$ local measurements, the GHZ state is decomposed differently \cite{guhne2007toolbox} and hence,
\begin{equation}\label{decomposition1}
\mathcal{W}_{N+1}=\frac{1}{2} \mathbb{I}-\frac{1}{2}M_z-\frac{1}{2N}\sum_{j=0}^{N-1} (-1)^j M_{\theta_j}^{\otimes{N}},
\end{equation}
where $\theta_j=\pi j/N$ and $M_{\theta_j}=\cos(\theta_j)\sigma_x+\sin(\theta_j)\sigma_y$. This decomposed EW involves a $\sigma_z$ basis measurement and $N$ measurements in the $\sigma_x$-$\sigma_y$ plane.

Note that the main difference between these two witnesses, Eq.~\eqref{witness1} and \eqref{decomposition1}, lies on the choice of LMSs and the corresponding coefficients. Inspired by this decomposition, instead of performing all the local obervables $M_{\theta_j}$ ($j=0,1\dots,N-1$), we only choose a subset of the measurements in the $\sigma_x$-$\sigma_y$ plane, and keep the observables in the $\sigma_z$ basis, i.e. $\{\ket{0}\bra{0}^{\otimes{N}}, \ket{1}\bra{1}^{\otimes{N}}\}$. Let $S\subset \{0,1,\dots,N-1\}$ be the set of chosen indices from the total $N$ indices and $|S|$ be the number of measurements in the $\sigma_x$-$\sigma_y$ plane. Similar to the decomposition in Eq.~\eqref{decomposition1}, we can define a Hermitian operator $\mathcal M_S$ as
\begin{equation}\label{measure}
\mathcal M_S=\ket{0}\bra{0}^{\otimes{N}}+\ket{1}\bra{1}^{\otimes{N}}+\frac{1}{C}\sum_{j\in S} (-1)^j M_{\theta_j}^{\otimes{N}},
\end{equation}
where $C$ is a weighting parameter which satisfies $|S|\le C\le 2|S|$ and can be optimized afterwards. Then, we can construct a generic witness $\mathcal W$ for an $N$-partite qubit system,
\begin{equation}
\begin{aligned}\label{witness3}
\mathcal{W}_S= \alpha_S \mathbb{I}- \mathcal M_S,
\end{aligned}
\end{equation}
where the parameter $\alpha_S$ depends on the set of the chosen measurement settings $S$ and coefficient $C$. The key challenge of constructing the genuine entanglement witness $\mathcal W$ is to calculate $\alpha_S$, Eq.~\eqref{eq:defalpha}, which, in general, is computationally hard for an arbitrary observable $\mathcal M_S$ \cite{guhne2009entanglement,doherty2004complete,gurvits2003classical}. In the following, we first formulate the mathematical problem and then present a method to efficiently estimate  an upper bound $\alpha_S^u$ of $\alpha_S$. Note that $\mathcal{W}_S= \alpha_S^u \mathbb{I}- \mathcal M_S$ is still a valid EW though sacrificing the performance of EW.

In order to evaluate Eq.~\eqref{eq:defalpha}, without losing generality, we only need to consider
pure bipartite separable states as general separable states can also be expressed as a mixture of pure bipartite separable states. Denote $\{k\}\{N-k\}$ to be a partition of the $N$ parties into a subset of $k$ parties and the complementary subset of $N-k$ parties. Denote $\mathcal H_{\{k\}}$ to be the Hilbert subspace of $k$ qubits. The invariance of a Hermitian operator $\mathcal M_S$ with respect to an arbitrary permutation of subsystem Hilbert spaces implies that to determine the biseparable bound, the detailed indexes falling into either subset $\{k\}$ or $\{N-k\}$ are irrelevant. Consider that a separable pure state $\ket{\phi}$ can be partitioned to $\{k\}\{N-k\}$ as $\ket{\phi}=\ket{\phi_A}\otimes \ket{\phi_B}$, where $\ket{\phi_A}\in \mathcal H_{\{k\}}$ and $\ket{\phi_B}\in \mathcal H_{\{N-k\}}$, and define the corresponding Hilbert subspace $\mathcal H_{\{k\}\{N-k\}}$. Define $f_{\{k\}\{N-k\}}$ to be the maximal quantum value of $\tr(\mathcal M_S\ket{\phi}\bra{\phi})$ over all biseparable $N$-qubit states in this partition,
\begin{equation}
\begin{aligned}
f_{\{k\}\{N-k\}}&=\max  \limits_{\ket{\phi}\in \mathcal H_{\{k\}\{N-k\}}}\ \tr(\mathcal M_S \ket{\phi}\bra{\phi}).
\end{aligned}
\end{equation}
Here, $f_{\{k\}\{N-k\}}$ can be obtained by performing the optimal measurements $\ket{\phi_A}\bra{\phi_A}$ on the $\mathcal H_{\{k\}} $ subsystem and finding the maximum eigenvalue of the rest matrix \cite{sperling2013multipartite},
\begin{equation}
\begin{aligned}
f_{\{k\}\{N-k\}}&=\max  \limits_{\ket{\phi_A}\in \mathcal H_{\{k\}}} \mathrm{eig}_{\max} \tr_{\{k\}}(\mathcal M_S\ket{\phi_A}\bra{\phi_A}),
\end{aligned}
\end{equation}
where $\mathrm{eig}_{\max}$ is the maximum eigenvalue of a matrix. The coefficient $\alpha_S$ defined in Eq.~\eqref{eq:defalpha} and \eqref{witness3} is the maximal value of $\tr(\mathcal M_S\sigma)$ over all biseparable $n$-qubit states including all the different bipartitions. Therefore, we have
\begin{equation} \label{fupp}
\alpha_S=\max_k f_{\{k\}\{N-k\}}.
\end{equation}
For each subset $S\subset \{0,1,\dots,N-1\}$ and parameter $C$, there is no straightforward way to obtain the value of $\alpha_S$ for the characterization of separable states is a computationally difficult problem \cite{doherty2004complete,gurvits2003classical}.

Generally, an upper bound of $\alpha_S$ can be achieved by semidefinite programming with symmetric extension \cite{doherty2004complete}. However, either such a general numerical procedure is inaccurate with exponential cost or it costs more time to give an accurate estimation for a large $N$. Furthermore, as we also need to search the optimal set of $S$, it makes the semidefinite programming method highly impractical. To resolve this problem, we propose a method to efficiently upper bound  $\alpha_S$ with computational cost $O(N)$.


\subsection{Estimating the upper bound for $\alpha_S$}
To estimate $\alpha_S$ defined in Eq.~\eqref{fupp}, we first give an analytical reduction and then  apply the numerical algorithm to compute the upper bound of $\alpha_S$. This numerical algorithm is efficient, which makes the optimisation of $S$ and $C$ possible in order to construct an optimal EW.

Define $M_{\mathrm{rest}}=\tr_{\{k\}}(\mathcal M_S\ket{\phi_A}\bra{\phi_A})$ and then
\begin{equation}
\begin{aligned}
M_{\mathrm{rest}} = & x\cdot\ket{0}\bra{0}^{\otimes{(N-k)}}+y\cdot\ket{1}\bra{1}^{\otimes{(N-k)}}
\\&+\frac{1}{C}\sum_{j\in S}(-1)^j\tr_{\{k\}}( M_{\theta_j}^{\otimes{N}}\ket{\phi_A}\bra{\phi_A}),
\end{aligned}
\end{equation}
where $x= \bra{\phi_A} (\ket{0}\bra{0}^{\otimes{k}})\ket{\phi_A}$, $y= \bra{\phi_A} (\ket{1}\bra{1}^{\otimes{k}})\ket{\phi_A}$  due to Eq.~\eqref{fupp}. This matrix $M_{\mathrm{rest}}$ is obtained by the Hermitian operator $\mathcal M_S$ after measuring $\ket{\phi_A}\bra{\phi_A}$ on first system $\{k\}$. Thus in order to estimate $\alpha_S$, we need to focus on the maximal eigenvalue of $M_{\mathrm{rest}}$.

We rewrite the expression of $M_{\mathrm{rest}}$ in a clearer form. Let function $l_1(b)$ for $b\in\{0,1\}^l$ be the bit count function, i.e., $l_1(b)=\sum_{i=1}^l b_i$. Then for a fixed $j\in S$, we have
\begin{equation}
\begin{aligned}
&\tr_{\{k\}}( M_{\theta_j}^{\otimes{N}}\ket{\phi_A}\bra{\phi_A})\\
&=z_j\cdot M_{\theta_j}^{\otimes{(N-k)}}\\
&=z_j\cdot \sum_{b\in\{0,1\}^{(N-k)}} e^{\mathrm i[l_1(b)-l_1(\neg b)]\theta_j}\ket{b}\bra{\neg b}\\
&=z_j\cdot \sum_{b\in\{0,1\}^{(N-k)}} e^{\mathrm i[2\cdot l_1(b)-(N-k)]\theta_j}\ket{b}\bra{\neg b}\\
&=z_j\cdot \left(
  \begin{array}{ccc}
  &&e^{-\mathrm i(N-k)\theta_j}\\
  & \iddots &\\
  e^{\mathrm i(N-k)\theta_j}& &
  \end{array}
\right),
\end{aligned}
\end{equation}
where $z_j= \bra{\phi_A} M_{\theta_j}^{\otimes{k}} \ket{\phi_A}$, and ``$\neg$'' means bitwise negation. The anti-diagonal elements are of the form $e^{\mathrm i t\theta_j}$ where $t=2\cdot l_1(b)-(N-k)$ whose range is $[-(N-k),N-k]$.

Combining all the measurement settings, we have
\begin{equation} \label{eq:Mrest}
M_{\mathrm{rest}}
=\left(
  \begin{array}{ccccc}
	  x & & &  &F_{z,k-N}\\
        & 0& & F_{z,k-N+1}&\\

        &  & \ddots &   &       \\
       & F_{z,N-k-1}& & 0 &\\
	  F_{z,N-k}& & & &y\\
  \end{array}
\right),
\end{equation}
where $F_{z,t}=\frac{1}{C}\sum_{j\in S} (-1)^j z_j e^{\mathrm it\theta_j}$ is defined as the anti-diagonal elements. It is not hard to see that $|F_{z,t}|\le |S|/C\le 1$.

\begin{lemma}
The maximal eigenvalue of $M_{\mathrm{rest}}$, defined in Eq.~\eqref{eq:Mrest}, is given by,
\begin{equation} \label{eq:maxeig}
\begin{aligned}
\lambda_{N-k}=\frac{x+y}{2}+\sqrt{|F_{z,N-k}|^2+\frac{(x-y)^2}{4}}.
\end{aligned}
\end{equation}
\end{lemma}

\begin{proof}
In order to calculate the eigenvalues of $M_{\mathrm{rest}}$, we perform elementary transformations on this matrix. There exists a sequence of elementary matrices $U_1U_2\cdots U_r$ such that $M_{\mathrm{rest}}'=(U_1U_2\cdots U_r)^{-1}M_{\mathrm{rest}}U_1U_2\cdots U_r$ and
\begin{equation}
M_{\mathrm{rest}}'=\left(
  \begin{array}{ccccc}
    x & F_{z,k-N} &&&\\
 F_{z,N-k}& y&&&\\
    &    & \ddots &&       \\
    &    & & 0&  F_{z,-t}     \\
    &    & &F_{z,t} &         0 \\
 \end{array}
\right).
\end{equation}
The diagonal elements of $M_{\mathrm{rest}}'$ are $2\times 2$ sub-matrices, in the form of
\begin{align*}
	\left(\begin{array}{cc}x &F_{z,k-N}\\ F_{z,N-k} &y\end{array}\right)\text{~ or ~}
		\left(\begin{array}{cc}0 & F_{z,-t}\\ F_{z,t} &0\end{array}\right).
\end{align*}
Then the eigenvalues of these matrices are
\begin{equation*}
\frac{x+y}{2}\pm\sqrt{|F_{z,N-k}|^2+\frac{(x-y)^2}{4}}\text{~ and ~}\pm|F_{z,t}|,
\end{equation*}
respectively, where $t\in [-(N-k-2),N-k-2]$. Since we only care about the maximal eigenvalue, denote
\begin{equation*}
\lambda_{N-k}=\frac{x+y}{2}+\sqrt{|F_{z,N-k}|^2+\frac{(x-y)^2}{4}}\text{~ and ~}\lambda_t=|F_{z,t}|.
\end{equation*} Moreover, $\lambda_t=|F_{z,t}|\le 1$ whereas if we pick $\phi_A$ such that $x=1$ and $y=0$, then
\begin{equation} \label{eq:maxlambda}
\begin{aligned}
\max\limits_{\ket{\phi_A}\in \mathcal H_{\{k\}}} \lambda_{N-k}\ge \frac{x+y}{2}+\sqrt{\frac{(x-y)^2}{4}}=1.
\end{aligned}
\end{equation}
So when we do the searching for the maximal eigenvalue, we can discard $\lambda_t$ and focus on $\lambda_{N-k}$.
\end{proof}

While the explicit expression for $\lambda_{N-k}$ is very complex, and hence instead of directly obtaining it, we find the upper bound for it in the following theorem.
\begin{theorem}
The coefficient $\alpha_S$ defined in Eq.~\eqref{fupp} can be upper bound by function $\lambda^u$,
\[\alpha_S\le \alpha_S^u=\max_k\max_{a_1,a_n} \lambda^u(a_1,a_n,k),\]
where 
\begin{equation}\label{eq:}
\begin{aligned}
\lambda^u(a_1,a_n,k)=\frac{x+y}{2}+\sqrt{(F_{z,N-k}^u)^2+\frac{(x-y)^2}{4}}.
\end{aligned}
\end{equation}
Here $F_{z,N-k}=\frac{1}{C}\sum_{j\in S} (-1)^j z_j e^{\mathrm i(N-k)\theta_j}$ is defined below Eq.~\eqref{eq:Mrest} and the maximization is over all possible partition number $k$ and measurement coefficients $a_1,a_n$
\end{theorem}

\begin{proof}
From the definition and the above lemma, we have
\begin{equation}
\begin{aligned}
\alpha_S=\max_{x,y,k}\lambda_{N-k}.
\end{aligned}
\end{equation}
In the expression of $\lambda_{N-k}$, the nontrivial part is $|F_{z,N-k}|$, which is the sum of $z_j$ with a phase. At first, we need to bound $z_j$.
Let $\ket{\phi_A}=(a_1,a_2,\cdots,a_n)^{\mathrm{T}}$ be a normalized state where $\mathrm{T}$ denotes the conjugate transpose and $n=2^k$, then we rewrite $z_j$ as follows,
\begin{equation}
\begin{aligned}
	z_j&=\bra{\phi_A}M_{\theta_j}^{\otimes{k}}\ket{\phi_A}\\
	&=a_1a_n^*e^{-\mathrm ik\theta_j}+a_1^*a_ne^{\mathrm ik\theta_j}+\sum_{q=2}^{n-1} a_q^* a_{n+1-q}\cdot e^{\mathrm it_q\theta_j},
\end{aligned}
\end{equation}
where $t_q\in[-k+2,k-2]$. The last term is a real value and satisfies
\begin{equation}
\begin{aligned}
	\sum_{q=2}^{n-1} a_q^* a_{n+1-q}\cdot e^{\mathrm it_q\theta_j}
	&\le \sum_{q=2}^{n-1} |a_qa_{n+1-q}| \\
	&\le \sum_{q=2}^{n-1} |a_q|^2 \\
	&= 1-|a_1|^2-|a_n|^2.
\end{aligned}
\end{equation}
Similarly, we also have
	$\sum_{q=2}^{n-1} a_q^* a_{n+1-q}\cdot e^{\mathrm it_q\theta_j}\ge -(1-|a_1|^2-|a_n|^2).$
Therefore we define
\begin{equation}
\begin{aligned}
z_j^u=\max|a_1a_n^*e^{-\mathrm ik\theta_j}+a_1^*a_ne^{\mathrm ik\theta_j}\pm(1-|a_1|^2-|a_n|^2)|
\end{aligned}
\end{equation}
and the absolute value for $z_j$ satisfies $|z_j|\le z_j^u$.
Thus for $F_{z,N-k}$, we have
\begin{equation}
\begin{aligned}\label{Fupperbound}
	|F_{z,N-k}|&=|\sum_{j\in S} (-1)^j z_j e^{\mathrm i(N-k)\theta_j}|/C\\
	&\le \max_{\mathrm{sign}}|\sum_{j\in S} \mathrm{sign}_j |z_j| e^{\mathrm i(N-k)\theta_j}|/C\\
	&\le \max_{\mathrm{sign}}|\sum_{j\in S} \mathrm{sign}_j z_j^u e^{\mathrm i(N-k)\theta_j}|/C,
\end{aligned}
\end{equation}
where $\mathrm{sign}=(\mathrm{sign}_j)_{j\in S}$ is a vector with $|S|$ elements taking values $\pm1$ on each element and the last inequality is proved in the Appendix~\ref{lemmas}.

Consequently, define
\begin{equation}
\begin{aligned}
F_{z,N-k}^u= \max_{\mathrm{sign}}|\sum_{j\in S} \mathrm{sign}_j z_j^u e^{\mathrm i(N-k)\theta_j}|/C
\end{aligned}
\end{equation} as the upper bound for $|F_{z,N-k}|$ and
\begin{equation}\label{eq:lambdau}
\begin{aligned}
\lambda^u(a_1,a_n,k)=\frac{x+y}{2}+\sqrt{(F_{z,N-k}^u)^2+\frac{(x-y)^2}{4}},
\end{aligned}
\end{equation}
is the upper bound for eigenvalues. Note that $x=|a_1|^2,\ y=|a_n|^2$, therefore $\lambda^u$ relies on two searching complex variables $a_1,a_n$ rather than the entire space $\mathcal{H}_{\{k\}}$. This reduces the searching space and leads to an upper bound for $f_{\{k\}\{N-k\}}$,
\begin{equation}
f_{\{k\}\{N-k\}}\le \max_{a_1,a_n} \lambda^u(a_1,a_n).
\end{equation}
By searching over the variable $k$, we can compute an upper bound for $\alpha_S$,
\[\alpha_S\le \alpha_S^u=\max_k\max_{a_1,a_n} \lambda^u(a_1,a_n).\]
\end{proof}

With the above theorem, we have transformed the original problem to a solvable optimization problem with only three variables $k,a_1,a_n$. Considering there are only three variables $k,a_1,a_n$,  $\alpha_S^u$ can be efficiently calculated in a numerical way with computational cost $O(N)$. The detailed algorithm for this optimization is shown in Algorithm~\ref{Algorithm}


\begin{algorithm}[H]
\begin{algorithmic}[1]
	\For{$s\gets1\ldots N$}
	\State Compute and store a $2^{s}\times s$ matrix $A^s$, whose rows are all possible $s$-length vectors taking values $\pm1$
	\EndFor
	\State Compute $\theta_j$ for all $j\in S$ and form a $|S|$-length vector $\vec \theta$
	\For{$k\gets1\ldots N-1$}
	\For{$\alpha\gets0\ldots \pi/2$ with step $\varepsilon$}
	\State $|a_1|\gets\cos\alpha$
	\For{$\beta\gets0\ldots \pi/2$ with step $\varepsilon$}
	\State $|a_n|\gets\sin\alpha\cos\beta$
	\For{$\phi\gets0\ldots \pi$ with step $\varepsilon$}
	\State $\vec{z^u}\gets\max|2|a_1||a_n|\cos(\phi-k\times \vec\theta)\pm (1-|a_1|^2-|a_n|^2)|$
	\State $\alpha_S\gets\max|A^{|S|}\cdot(\vec{z^u}\odot e^{\mathrm i(N-k)\vec\theta})|/C$
	\State Compute $\lambda^u$ using $\alpha_S$, record it
	\EndFor
	\EndFor
	\EndFor
	\EndFor
	\State \textbf{return} $\alpha_S\gets \max \lambda^u$ in records
\end{algorithmic}
	\caption{Compute upper bound for $\alpha_S$ with inputs $S,C$ and step size $\varepsilon$ and output $\alpha_S$. $\odot$ is pairwise multiplication, $\cdot$ is matrix multiplication}\label{Algorithm}
\end{algorithm}

\section{Robustness against efficiency}\label{Sec:robust}

In this section, we apply the result in the previous sections to search for the optimal measurement settings for a given number of LMSs $|S|$, and study about the relationship between robustness and efficiency.

In practice, a generated entangled state is usually not the target pure state $\ket{\mathrm{GHZ}}$, but a mixture of the target pure state with noises. For simplicity, we could regard the noise as white noise, thus the generated state is
\begin{equation}
\rho=(1-p)\ket{\mathrm{GHZ}}\bra{\mathrm{GHZ}}+p\mathbb{I}/2^N.
\end{equation}
Applying the witness $\mathcal W= \alpha_S^u \mathbb{I}- \mathcal M_S$, a mixed state $\rho$ will lead to
\begin{equation}
\tr(\mathcal W\rho)=\alpha_S^u-(1-p)(1+|S|/C)-p/2^{N-1}.
\end{equation}
When $N$ tends to infinity, the maximal tolerable white noise parameter $p$ is roughly $1-\alpha_S/(1+|S|/C)$. Here, we regard the robust witness which can tolerate the maximal white noise as the optimal one. Thus for a given number of LMSs $|S|$, we need to compare the tolerable white noise for each selection of measurement settings and find the optimal one.

In order to find the optimal choice of measurement settings, we numerically calculate the corresponding $\alpha_S^u$ and obtain the tolerable white noise for every possible $S$ with $N=2,\dots,15$. We take $N=5,10,15$ as examples to show our results in details. The maximal tolerable white noise for different $|S|$ is shown in Fig.~\ref{Fig:res}.


\begin{figure*}[t]\centering 
\subfigure[]{\includegraphics[width=0.85\columnwidth]{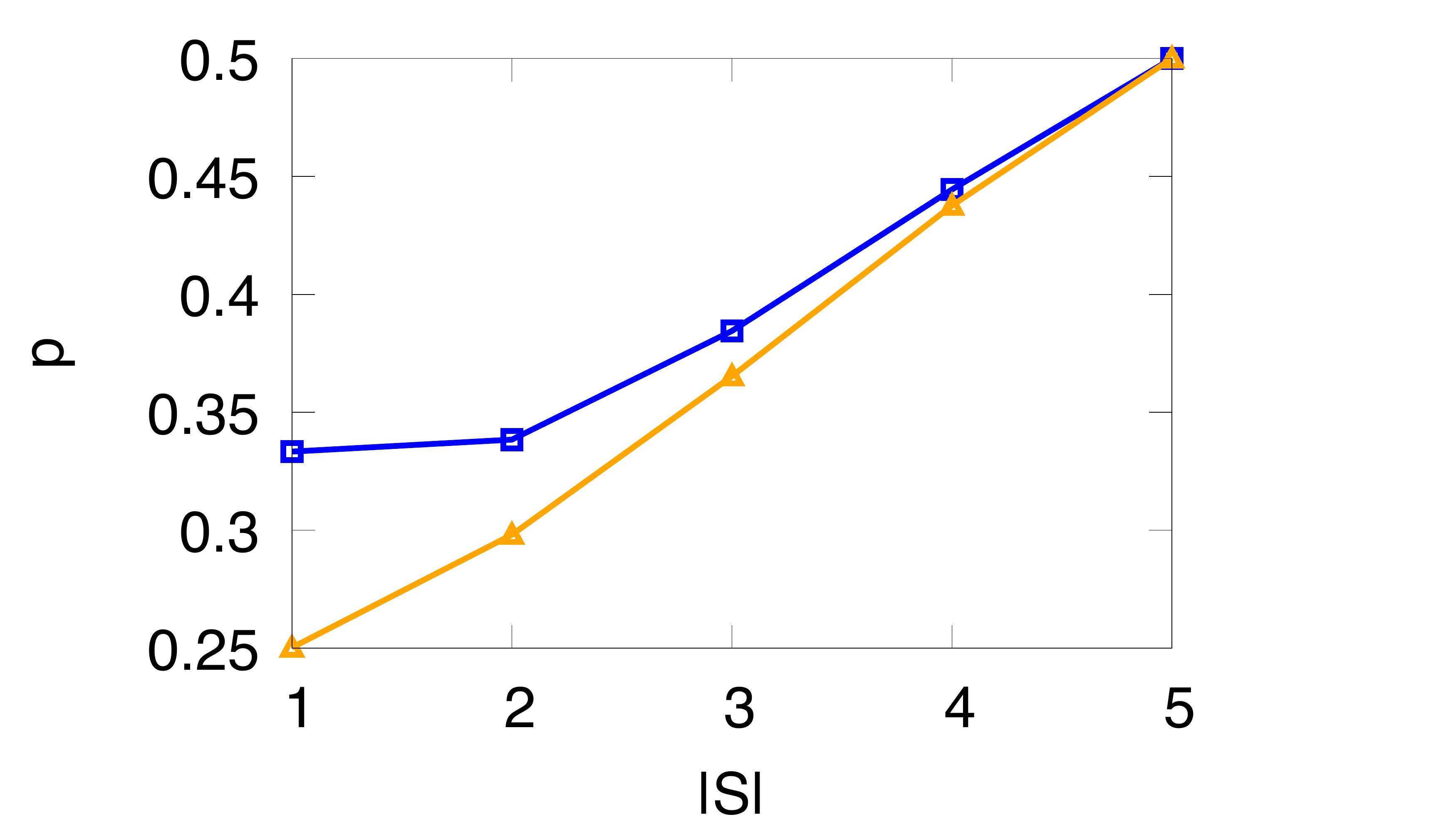}}\hspace{-8em}
\subfigure[]{\includegraphics[width=0.85\columnwidth]{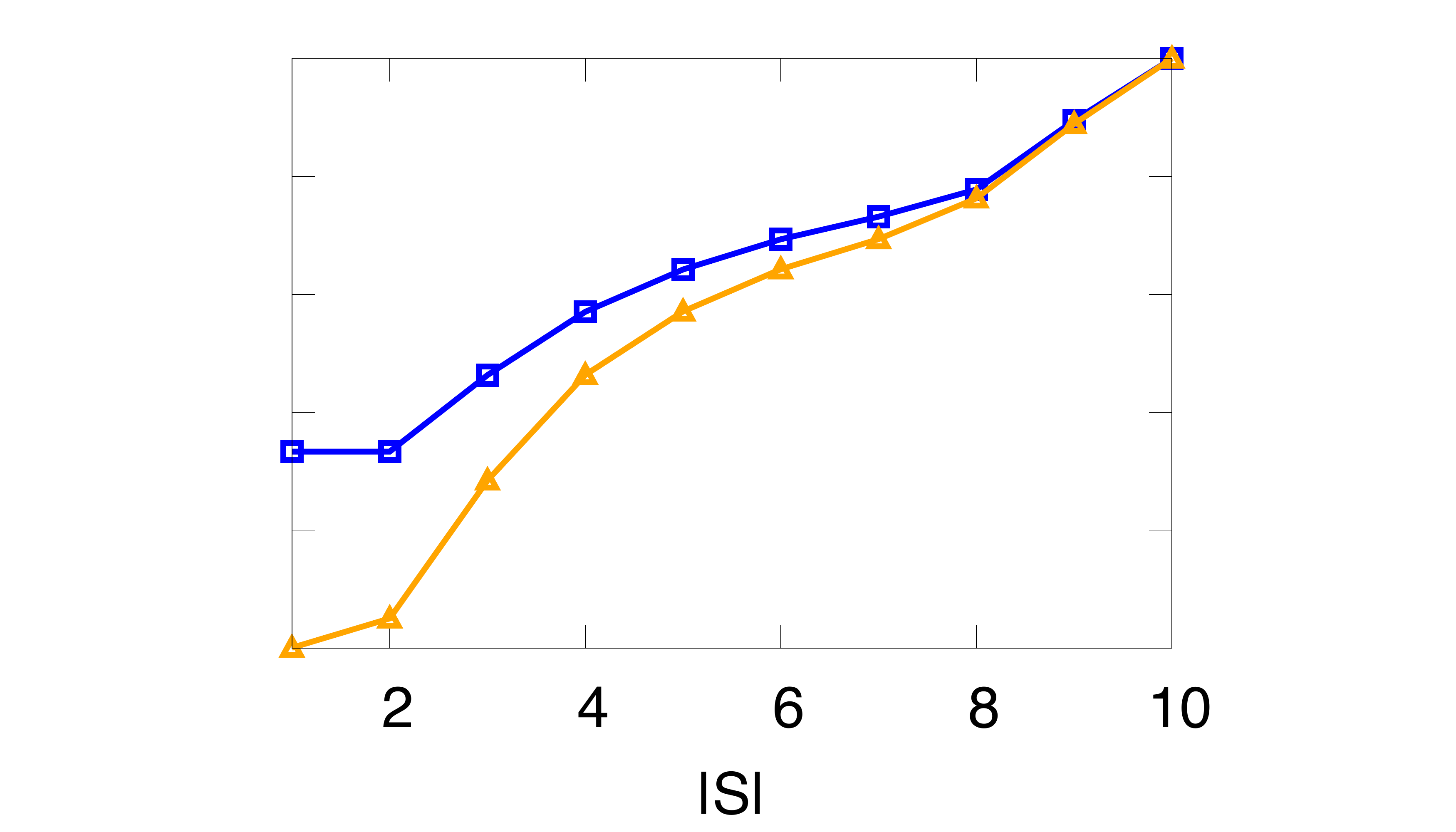}}\hspace{-8em}
\subfigure[]{\includegraphics[width=0.85\columnwidth]{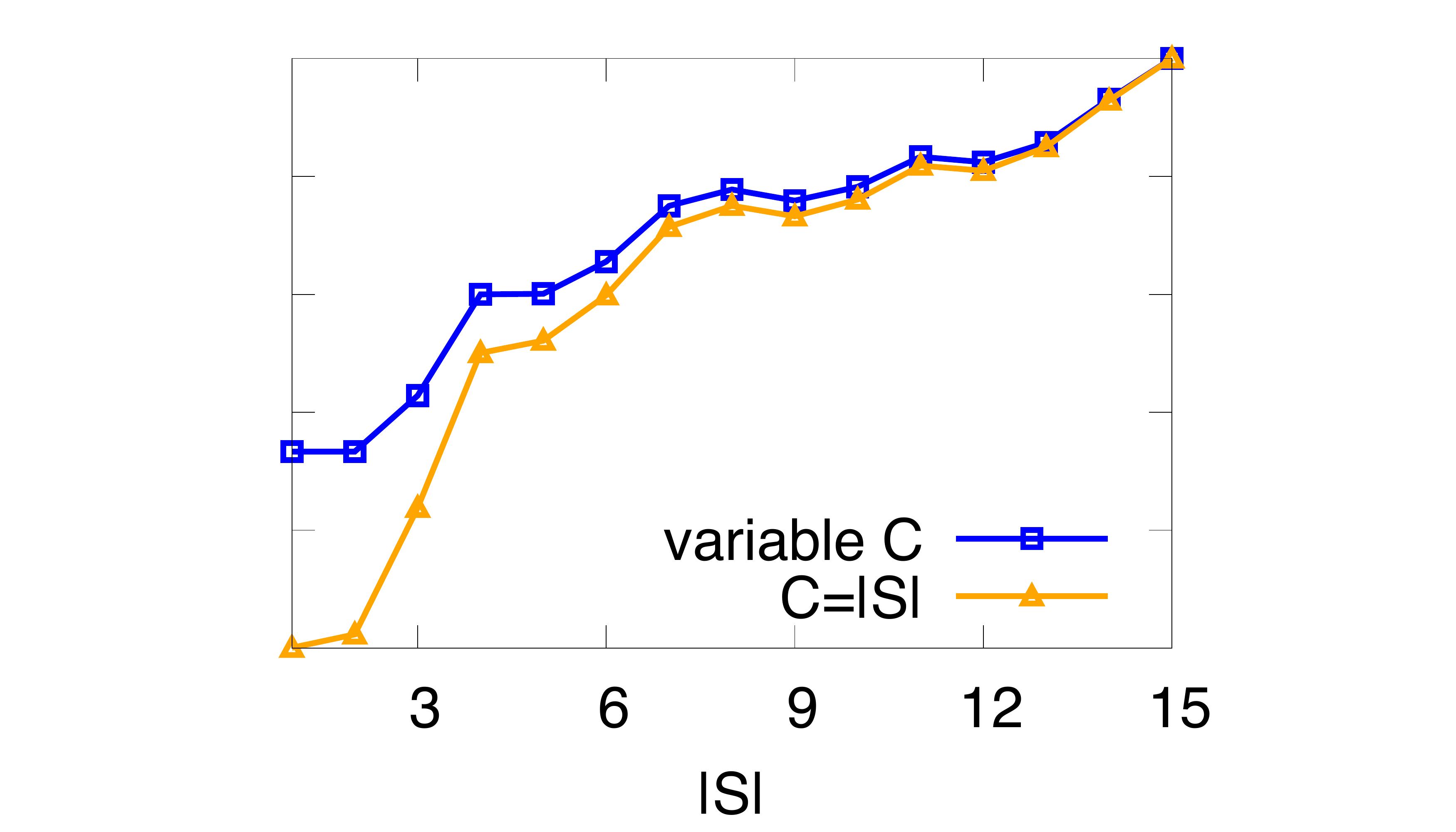}} 
\caption{The maximal tolerable white noise $p$ versus the number of LMSs for (a) $N=5$ (b) $N=10$ (c) $N=15$} \label{Fig:res} 
\end{figure*}

There is still another parameter $C$ which needs to be optimized. Here, we show two different cases, an optimized $C$ and a fixed $C=|S|$ . For ``variable $C$'' in Eq.~\eqref{measure}, we optimize $C\in[|S|,2|S|]$ (blue line in our figures) whereas for ``$C=|S|$'' we fix parameter $|C|$ (orange line in our figures). To be more specific, in order to give an instruction for the related experiments, we also list the optimal choice of measurement settings  and optimal $C$ for different $|S|$, $N=5,10,15$ in Table \ref{tableN5},~\ref{tableN10},~\ref{tableN15}, respectively. In these tables, the measurement angle $[j]$ denotes the local observable $M_{\theta_j}$ $(j=0,1\dots,N-1)$ in Eq.~\eqref{measure}.

\begin{table}[bht]
\centering
\begin{tabular}{cccc}
\hline
 $|S|$  & $p$ & Measurement angles & $C_{\textrm{opt}}$
\\ \hline
1& 0.333 & [0] & 2\\
2 & 0.338 & [1, 4]&3\\
3 & 0.384 & [1, 2, 3] & 4\\
4 & 0.444 &[0, 1, 2, 3]&5\\
5  & 0.500 & [0, 1, 2, 3, 4]	 &5 \\
\hline
\end{tabular}
\caption{Summary of optimal choice of measurement settings for for $N=5$. $C_{\textrm{opt}}$ is the optimal $C$.}\label{tableN5}
\end{table}

\begin{table}[tbh]\label{tableN10}
\centering
\begin{tabular}{cccc}
\hline
 $|S|$  & $p$ & Measurement angles & $C_{\textrm{opt}}$
\\ \hline
1	&0.333 &	[0]&	2\\
2	&0.333&	[0, 1]	&4\\
3	&0.366&	[1, 3, 4]	&5\\
4	&0.393&	[0, 1, 2, 7]	&6\\
5	&0.411&	[1, 2, 4, 5, 6]	&7\\
6	&0.423&	[1, 2, 3, 4, 6, 7]	&8\\
7	&0.433&	[1, 2, 3, 4, 5, 6, 8]	&9\\
8	&0.444&	[0, 1, 2, 3, 4, 5, 6, 7]	&10\\
9	&0.474&	[0, 1, 2, 3, 4, 5, 6, 7, 8]&	10\\
10	&0.500	&[0, 1, 2, 3, 4, 5, 6, 7, 8, 9]	&10\\
\hline
\end{tabular}
\caption{Summary of optimal choice of measurement settings for $N=10$. $C_{\textrm{opt}}$ is the optimal $C$.}\label{tableN10}
\end{table}

\begin{table}[htb]\label{tableN15}
\centering
\begin{tabular}{cccc}
\hline
 $|S|$  & $p$ & Measurement anlges &  $C_{\textrm{opt}}$
\\ \hline
1	&0.333&	[0]&	2\\
2	&0.333&	[0, 10]	&4\\
3	&0.357&	[1, 4, 5]	&5\\
4	&0.400&	[0, 1, 3, 7]&	6\\
5	&0.400&	[0, 1, 2, 7, 13]&	7\\
6	&0.414&	[1, 2, 4, 5, 10, 12]&	8\\
7	&0.438&	[0, 1, 2, 4, 5, 8, 10]&	9\\
8	&0.444&	[0, 1, 2, 3, 5, 7, 8, 11]&	10\\
9	&0.440&	[1, 2, 3, 4, 5, 7, 8, 10, 12]&	11\\
10	&0.446&	[0, 1, 2, 3, 4, 6, 7, 9, 10, 14]	&12\\
11	&0.458&	[0, 1, 2, 3, 4, 5, 6, 8, 9, 10, 12]	&13\\
12	&0.456&	[0, 1, 2, 3, 4, 5, 6, 9, 10, 12, 13, 14]&	14\\
13	&0.464&	[0, 1, 2, 3, 4, 5, 6, 7, 10, 11, 12, 13, 14]&	15\\
14	&0.483&	[0, 1, 2, 3, 4, 5, 6, 7, 8, 10, 11, 12, 13, 14]&	15\\
15	&0.500&	[0, 1, 2, 3, 4, 5, 6, 7, 8, 9, 10, 11, 12, 13, 14]&	15\\
\hline
\end{tabular}
\caption{Summary of optimal choice of measurement settings for $N=15$. $C_{\textrm{opt}}$ is the optimal $C$.}\label{tableN15}
\end{table}

In practice, it is not hard to experimentally calibrate the white noise level $p$ for a generated state. They could find the most efficient witness for a given $p$ according to these tables. Moreover, when $|S|=1$, the obtained optimal measurement settings is $M_{\theta_0}=\sigma_x$, which is consistent with the witnesses proposed in previous works \cite{PhysRevLett.94.060501,PhysRevLett.117.210504}.

\section{Discussion} \label{Sec:Conclusion}
In this work, we study the genuine entanglement detection of GHZ-like states. With different LMSs, we propose a family of EWs that can tolerate different levels of white noises. We optimize the number of measurement settings for a given noise rate. As the GHZ-like state is an important multipartite entangled state, our results can be directly applied in practice, with realistic state preparation. As the number of partitions $N$ increases, completely performing $N$ measurement settings for an entanglement witness cost too much efforts. Our work solve this practical problem and provide an alternative way of entanglement witness. In practice, a suitable witness can be chosen according to the noise level of the practical system. Furthermore, the proposed optimal witness is also helpful to improve the performance of measurement-device-independent EW \cite{Branciard13,Yuan14,
nawareg2015experimental,Buscemi12,Zhao16,PhysRevA.93.042317}.

The technique used in this work still has the room for improvement. For instance, in finding more efficient way to search for the optimal witness and in estimating $\alpha_S^u$ with a higher accuracy.
While searching for the optimal choice of measurements, we find that many choices have the same $\alpha_S^u$ and tolerable noise, but for simplicity we only list one of the optimal choices in the tables. Because the symmetry of the $\sigma_x-\sigma_y$ plane, there exists the simplification for searching optimal $S$ and we give the equivalence of $S$ in Appendix \ref{simplication}. The property of the symmetry of $\sigma_x-\sigma_y$ plane could be applied more systematically to simplify the searching procedure.

Note that from the figures, roughly speaking, we could see that more measurement settings $|S|$ will lead to a more robust witness. However, it is not strict because the tolerable noise $p$ is not monotonically increasing with $|S|$. Sometimes, more measurement settings will not have benefits. This stems from the fact that our estimated upper bound $\alpha_S^u$ is not always tight. The calculation of an accurate $\alpha_S$ in an analytical way may be helpful for improving the performance and reducing the complexity for finding the optimal choice.

Moreover, the techniques proposed in this work can be extended to other symmetric states \cite{zhou19symmetric} or high dimension GHZ-like state. Constructing genuine entanglement witness family in other forms is also an interesting direction. The other types of genuinely entangled states, such as graph states \cite{hein2004multiparty} and Dicke state \cite{dicke1954coherence}, may require a different method of constructing witnesses. Besides the traditional entanglement witness, the similar technique can also be generalized to other entanglement detection tasks, i.e., entanglement structure detection \cite{PhysRevX.8.021072}.

\section*{ACKNOWLEDGMENTS}
We acknowledge Yeong-Cherng Liang for the insightful discussions. This work was supported by the National Natural Science Foundation of China Grants No.~11875173 and No.~11674193, and the National Key R\&D Program of China Grant No.~2017YFA0303900. X.Y.~acknowledge EPSRC grant EP/M013243/1.

\appendix
\section{Proof of Eq.~\eqref{Fupperbound}}\label{lemmas}
	Here we regard the complex numbers as two-dimensional vectors in the complex plane. We use $\inn{z_1}{z_2}$ to denote the inner product for two complex numbers $z_1,z_2$. That is,
\begin{equation}
    	\inn{z_1}{z_2}=|z_1||z_2|\cos(\phi_1-\phi_2)=x_1x_2+y_1y_2
\end{equation}
 where $z_j=|z_j| e^{\mathrm i\phi_j}=x_j+\mathrm i y_j,\ j=1,2$.
	
	For a fixed set of indices $S$ and any $j\in S$, let $\mathrm{sign}_j\in\{\pm1\}$ and $\mathrm{sign}=(\mathrm{sign}_j)_{j\in S}$ be a bivalent vector  . Let vectors (complex values) $z_j=|z_j| e^{\mathrm i\phi_j}$, whose modules have bound $|z_j|\le z_j^u$. Define the sum $H_{\mathrm{sign}}=\sum_{j\in S} \mathrm{sign}_j z_j$. And let $\mathrm{sign}^*$ be the vector that maximize the module of $H$. Then $H_{\mathrm{sign}^*}$ is the maximal $H$ and we have the following lemma.
	\begin{lemma}
		$\inn{\mathrm{sign}^*_j z_j}{H_{\mathrm{sign}^*}}\ge0$ for all $j\in S$.
		\label{lem1}
	\end{lemma}
\begin{proof}
	We prove this lemma by contradiction. Assume there exists $i\in S$ such that $\inn{\mathrm{sign}^*_i z_i}{H_{\mathrm{sign}^*}}<0$, then define $\mathrm{sign}'$ as follows,
\begin{equation}
   \mathrm{sign}'_j=\left\{
		\begin{array}{lr}
			-\mathrm{sign}^*_j&j=i\\
			\mathrm{sign}^*_j&j\neq i
		\end{array}\right.
\end{equation}		
which means $H_{\mathrm{sign}'}=\sum_{j\in S}\mathrm{sign}'_j z_j=H_{\mathrm{sign}^*}-2\mathrm{sign}^*_i z_i$. Then
\begin{equation}
\begin{aligned}
&\inn{H_{\mathrm{sign}'}}{ H_{\mathrm{sign}'}}\\
			=&\inn{H_{\mathrm{sign}^*}}{ H_{\mathrm{sign}^*}}-4\inn{\mathrm{sign}^*_i z_i}{H_{\mathrm{sign}^*}}+4 \inn{\mathrm{sign}^*_i z_i}{\mathrm{sign}^*_i z_i}\\
			>&\inn{H_{\mathrm{sign}^*}}{ H_{\mathrm{sign}^*}}.
	    \end{aligned}
	\end{equation}

		It follows that $|H_{\mathrm{sign}'}|^2>|H_{\mathrm{sign}^*}|^2$. And it contradicts to the definition of $\mathrm{sign}^*$.
\end{proof}
	Applying this lemma we can prove another lemma which uses the upper bound of complex values $z_j$ to bound their sum $H_{\mathrm{sign}}$. Define $H_{\mathrm{sign}}^u=\sum_{j\in S} \mathrm{sign}_j z_j^u e^{\mathrm i\phi_j}$ and $\mathrm{sign}^{u}$ is the vector that
maximize the module of $H_{\mathrm{sign}}^u$. Then we have the next lemma.
	\begin{lemma}
		$|H_{\mathrm{sign}^*}|\le|H_{\mathrm{sign}^u}^u|$, that is,
\begin{equation}
    \max_{\mathrm{sign}}|\sum_{j\in S} \mathrm{sign}_j |z_j| e^{\mathrm i\phi_j}|\le \max_{\mathrm{sign}}|\sum_{j\in S} \mathrm{sign}_j z_j^u e^{\mathrm i\phi_j}|.
\end{equation}
		\label{lem2}
	\end{lemma}
	\begin{proof}
		Denote $\Delta z_j=z_j^u e^{\mathrm i\phi_j}-z_j=(z_j^u-|z_j|)e^{\mathrm i\phi_j}$ and $\Delta H =H_{\mathrm{sign}^*}^u-H_{\mathrm{sign}^*}=\sum_{j\in S}\mathrm{sign}_j^*\Delta z_j$. Observe that $\inn{\mathrm{sign}_j^*\Delta z_j}{H_{\mathrm{sign}^*}}\ge 0$ for all $j\in S$ since we have $\inn{\mathrm{sign}^*_j z_j}{H_{\mathrm{sign}^*}}\ge 0$ by lemma~\ref{lem1} and $z_j,\ \Delta z_j$ have the same direction or $\Delta z_j=0$. Hence $\inn{\Delta H}{H_{\mathrm{sign}^*}}\ge0$. It follows that
\begin{equation}
\begin{aligned}
&\inn{H_{\mathrm{sign}^*}^u}{H_{\mathrm{sign}^*}^u}\\
				=&\inn{H_{\mathrm{sign}^*}}{ H_{\mathrm{sign}^*}}+2\inn{\Delta H}{H_{\mathrm{sign}^*}}+\inn{\Delta H}{\Delta H}\\
				\ge& \inn{H_{\mathrm{sign}^*}}{ H_{\mathrm{sign}^*}}.
		\end{aligned}
	\end{equation}
Hence
\begin{equation}
|H_{\mathrm{sign}^*}|\le|H_{\mathrm{sign}^*}^u|\le|H_{\mathrm{sign}^u}^u|.
\end{equation}
	\end{proof}

\section{Equivalence of local measurement setting $S$ }\label{simplication}
We denote $\mathcal R_Z$ as the set of rotations by z-axis and $\Sigma$ as the set of all $S$. $U\in\mathcal R_Z$ as a fixed transformation. For two distinct sets $S,S'\in\Sigma$ with the same size $|S|$, if there exists a matching of elements in $S$ and $S'$ such that for any $(j,j'),j\in S,j'\in S'$ in this matching, we have $(-1)^{(j-j')}=1$ and $M_{\theta_{j'}}=U M_{\theta_j}U^{\dagger}.$ We call $S$ and $S'$ {\bf equivalent by} $U$.

The operators defined in Eq.~\eqref{witness3} with $S$, $S'$ are $\mathcal M_S$ and $\mathcal M_S'$. Then for any state $\ket{\phi}$, there exists a state $\ket{\phi'}$ such that $\bra{\phi}\mathcal M_S\ket{\phi}=\bra{\phi'}\mathcal M_S'\ket{\phi'}$.

	\begin{proof}
	Let $\ket{\phi'}=U\ket{\phi}$. Then for any $(j,j')$ in the matching,
	\begin{equation}
		\bra{\phi'}M_{\theta_{j'}}\ket{\phi'}=\bra{\phi}U^{\dagger} U M_{\theta_j}U^{\dagger} U\ket{\phi}=\bra{\phi}M_{\theta_j}\ket{\phi}.
	\end{equation}
	In addition, $U\in\mathcal R_Z$ so the z-axis observable remains. Thus $\bra{\phi}M\ket{\phi}=\bra{\phi'}M'\ket{\phi'}$ follows.
	\end{proof}





\bibliography{bibGME}



\end{document}